\documentclass[conference]{IEEEtran}
\IEEEoverridecommandlockouts
\usepackage{cite}
\usepackage{amsmath,amssymb,amsfonts}
\usepackage{algorithm}
\usepackage{bm}
\usepackage{algorithmic}
\usepackage{graphicx}
\usepackage{url}
\usepackage{textcomp}
\usepackage{amsthm}
\newtheorem{theorem}{Theorem}
\newtheorem{remark}{Remark}
\usepackage{ulem}

\usepackage[T1]{fontenc}

\usepackage{xcolor}
\def\BibTeX{{\rm B\kern-.05em{\sc i\kern-.025em b}\kern-.08em
    T\kern-.1667em\lower.7ex\hbox{E}\kern-.125emX}}

\begin{document}

\title{CoCoFL: Continual Computing for Federated Learning over Intermittent Satellite-Ground Links}

\author{\IEEEauthorblockN{Yun Shen\textsuperscript{$1$}, Kun Guo\textsuperscript{$1*$}, Xi Yang\textsuperscript{$1$}, Yaoqi Liu\textsuperscript{$2$}, Yisheng Zhao\textsuperscript{$3$}, and Wei Feng\textsuperscript{$4$}\\
\IEEEauthorblockA{\textsuperscript{$1$}School of Information and Electronic Engineering, East China Normal University, Shanghai, China\\
\textsuperscript{$2$}Institute of Computing Technology, Chinese Academy of Sciences, Beijing, China\\
\textsuperscript{$3$}Hangzhou Institute for Advanced Study, University of Chinese Academy of Sciences, Hangzhou, China\\
\textsuperscript{$4$}Department of Electronic Engineering, Tsinghua University, Beijing, China\\
\textsuperscript{$*$}Correspondence: kguo@cee.ecnu.edu.cn
\thanks{This research was supported by the National Natural Science Foundation of China, under Grant 62571191 and 62425110.}
}
}}

\maketitle

\begin{abstract}

Low earth orbit (LEO) satellite constellations enable geographically distributed ground devices to collaboratively train a global model via federated learning (FL) without sharing raw data, with applications in environmental monitoring and disaster prediction.
However, in satellite-assisted FL scenarios, intermittent satellite–ground links allow only a subset of devices to participate in global aggregation within each visibility window, leaving unscheduled devices idle and their local computational and data resources underutilized.
Under partial device participation, data heterogeneity among devices may bias the global model toward certain devices, thereby deteriorating learning performance.
In this regard, we propose a continual computing based federated learning framework, referred to as CoCoFL, in which scheduled devices participate in the global model aggregation, while unscheduled devices continue updating their local models taking into account model staleness.
Guided by the convergence analysis of CoCoFL and subject to visible-window-related time constraints, we jointly optimize the device scheduling and the number of local epochs for scheduled and unscheduled devices. Experimental results demonstrate that CoCoFL achieves faster convergence, lower training loss, and higher test accuracy compared with baselines.
\end{abstract}

\begin{IEEEkeywords}
Continual computing, federated
learning, intermittent communication
\end{IEEEkeywords}

\section{Introduction}

Low earth orbit (LEO) satellite constellations offer global connectivity for ground devices deployed in remote areas, which collect large and highly distributed datasets for applications such as environmental monitoring and disaster prediction \cite{Wu2025iotj}.
However, transmitting these datasets to cloud servers via satellite–ground links introduces significant communication overhead. Federated learning (FL), a distributed learning paradigm, allows ground devices to collaboratively train a global model without sharing raw data \cite{mcmahan2017communication}. In this regard, satellite-assisted FL, where LEO satellites act as parameter servers, has attracted increasing attention.


Recent studies have explored satellite-assisted FL~\cite{chen2025tmc, fang2023twc, Han2024JSAC}, accounting for satellite mobility and constellation topology to reduce communication overhead and mitigate the effects of data heterogeneity. Nevertheless, most existing works are designed for continuous connectivity between satellites and ground devices, neglecting the scenario of intermittent satellite–ground links. In such cases, ground devices can communicate with satellites only within limited visibility windows~\cite{Lin2025TMC}, which restricts the number of devices scheduled for global aggregation in each round. In conventional synchronous FL, scheduled devices perform local training and participate in global aggregation, while unscheduled devices remain idle, resulting in underutilization of their local computational and data resources. Moreover, heterogeneity in local dataset sizes and class distributions across devices leads to uneven local training delays and performance degradation ~\cite{XuefeiTMC,Luo2022Infocom}.

Semi-asynchronous FL aggregates model updates based on arrival thresholds or deadlines and is particularly well suited for scenarios with partial device participation in each training round. In addition, it more effectively accommodates variations in local computation delays arising from device heterogeneity compared to synchronous FL. Nevertheless, existing semi-asynchronous FL approaches generally assume continuous connectivity between satellites and ground devices, allowing devices to upload model parameters to the parameter server immediately after completing local training~\cite{Ma2021JSAC, You2023TWC, Chen2025Tcom}. This assumption limits the direct applicability of conventional semi-asynchronous FL to satellite-assisted FL scenarios with intermittent satellite–ground links. 

To address these challenges, we propose CoCoFL, a continual computing based federated learning framework over intermittent satellite–ground links.
In CoCoFL, unscheduled devices continue local training under model staleness constraints, temporarily caching their updates for transmission during subsequent satellite passes.
This design parallelizes local updates of unscheduled devices with global aggregation of scheduled devices to fully exploit local computational and data resources.
Furthermore, we conduct a convergence analysis and develop a two-level optimization algorithm within CoCoFL guided by its results to jointly optimize device scheduling and local epochs.
Extensive simulations demonstrate that CoCoFL significantly accelerates convergence and improves both training loss and test accuracy compared to baselines.

The remainder of this paper is organized as follows.
Section~II describes the system model.
Section~III presents the convergence analysis and problem formulation.
Section~IV introduces the proposed learning acceleration algorithm, followed by the experimental results in Section~V.
Finally, Section~VI concludes the paper.

\section{System Model}
In this section, we elaborate on the learning workflow of the proposed CoCoFL framework, followed by the corresponding communication and computing models for further optimization of the learning process.

\subsection{Learning Workflow of the Proposed CoCoFL \label{sub:CoCoFL}}

We consider an FL scenario in which multiple ground devices are deployed in remote areas, with LEO satellites sequentially serving as parameter servers. To clarify our contribution, we assume at most one LEO satellite acts as the parameter server during each visibility window. As satellites move along their orbits, the parameter server role is handed over to the next satellite entering the region.

Let the set of ground devices be denoted by 
$\mathcal{I}=\{1,2,\ldots,I\}$, where $I$ is the total number of devices. 
Each device $i\in\mathcal{I}$ holds a local dataset $\mathcal{D}_i$ with size 
$D_i = |\mathcal{D}_i|$, and the total number of samples is 
$D = \sum_{i\in\mathcal{I}} D_i$. The goal of FL is to minimize the weighted sum of local loss functions across all devices: 
\begin{equation}
F(\boldsymbol{w}) = \sum_{i \in \mathcal{I}} \beta_i F_i(\boldsymbol{w}),
\label{eq:global_loss}
\end{equation}
where $\beta_i = D_i / D$ is the fraction of the local dataset and $F_i(w)$ is the local loss, defined as

\begin{equation}
F_i(\boldsymbol{w}) = \frac{1}{D_i} \sum_{d_{i,j}\in\mathcal{D}_i} f(\boldsymbol{w}; d_{i,j}),
\end{equation}
where $f(\boldsymbol{w}; d_{i,j})$ denotes the loss function value evaluated on the data sample $d_{i,j}$.


\begin{figure}[t]
\centerline{\includegraphics[scale=0.16]{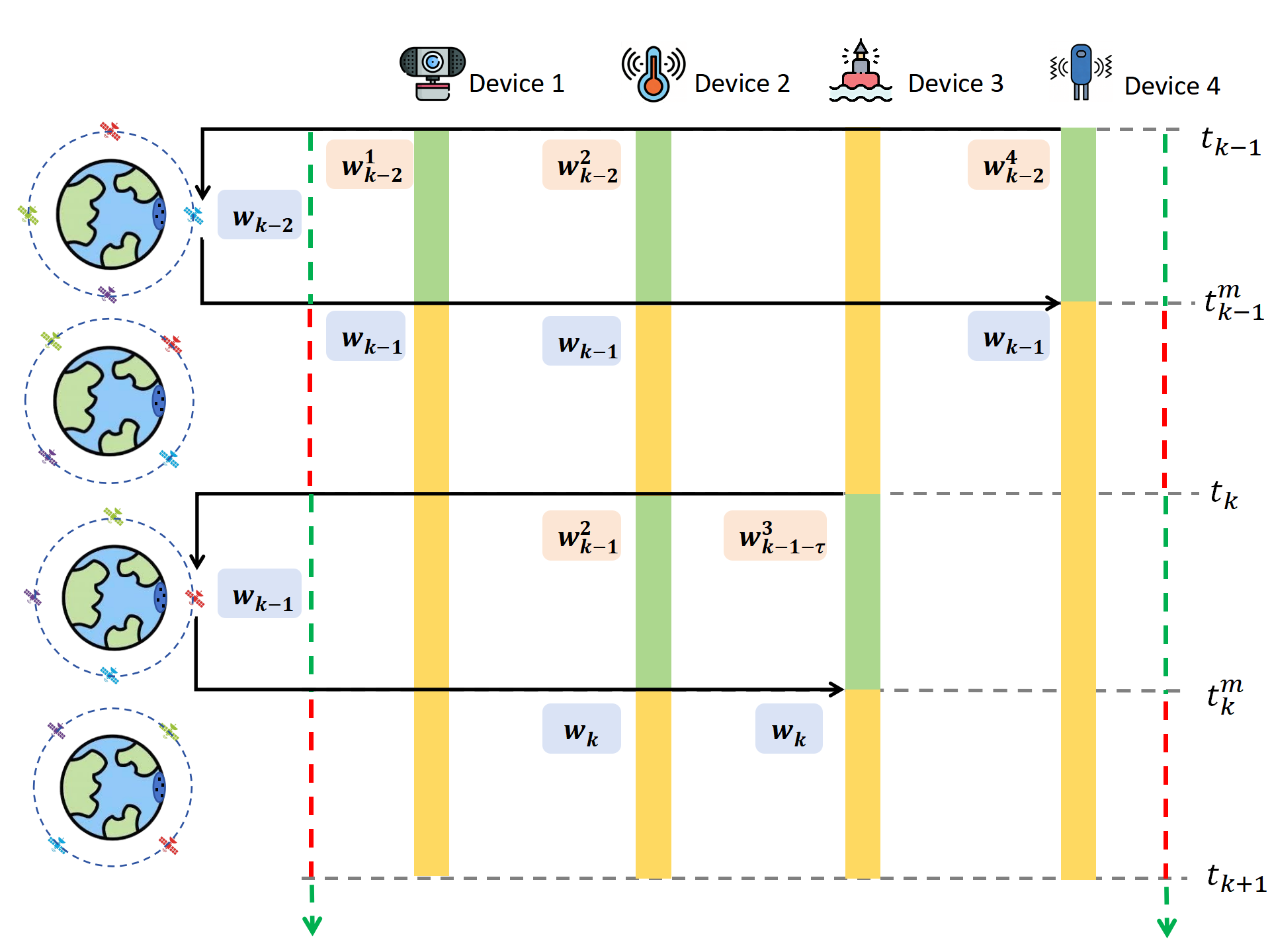}}
\caption{Illustration of CoCoFL: green solid lines denote communication phases, and yellow solid lines denote computation phases; $w_{k}^i$ means the local model update of device $i$ based on the global model $w_{k}$.}
 \vspace{-0.5cm}
\label{fig:procedure}
\end{figure}

To minimize \eqref{eq:global_loss} over intermittent satellite-ground links, we propose the CoCoFL framework to train a global model with $K$ communication rounds. As illustrated in Fig.~\ref{fig:procedure}, the duration of the $k$-th round spans from $t_k$ to $t_{k+1}$, corresponding to the interval between two consecutive satellite passes. The interval $[t_k, t_k^m)$ represents the visibility window, while $[t_k^m, t_{k+1})$ denotes the invisible window. The $k$-th round proceeds as follows.

\subsubsection{\textbf{Device Scheduling and Model Uploading}}
During the visibility window $[t_k,t_k^m)$,
ground devices can communicate with the satellite.
Due to limited visibility windows, only a subset of devices can
upload their local model updates for global aggregation.
We therefore introduce the staleness coefficient $\tau_k^i$ to 
represent the number of consecutive rounds up to round $k$, during which
device $i$ has not been scheduled for global aggregation.
Let $\mathcal I_k \subseteq \mathcal I$ denote the set of devices scheduled in round $k$, and then we give $\tau_k^i$ below:
\begin{equation}
\tau_k^i =
\begin{cases}
0, & i\in\mathcal I_k,\\
\tau_{k-1}^i+1, & i\in\mathcal I\setminus\mathcal I_k.
\end{cases}
\end{equation}
That is, the staleness coefficient of one device is reset to zero if it is scheduled for global aggregation in the current round; otherwise, it is increased by one. Furthermore, the scheduled device $i$ uploads its local model update $\boldsymbol{w}^i_{k-1-\tau_{k-1}^i}$ to the visible LEO satellite. 

\subsubsection{\textbf{Global Model Update}}
After receiving the local model updates from all scheduled devices, the satellite aggregates them together with the global model from round $k-1$, denoted by $\boldsymbol{w}_{k-1}$, to update the global model as follows:

\begin{equation}
\label{aggregate}
\boldsymbol{w}_k=\sum_{i\in\mathcal I \setminus \mathcal{I}_k}\beta_i\boldsymbol{w}_{k-1}
+\sum_{i\in\mathcal I_k}\beta_i
\boldsymbol{w}^i_{k-1-\tau_{k-1}^i}.
\end{equation}
The global model from the previous round is incorporated to mitigate fluctuations in the updated global model when only a small number of devices are scheduled in the current round.

\subsubsection{\textbf{Global Model Broadcasting}}
The updated global model $\boldsymbol{w}_k$ is
broadcast to the scheduled devices before the end of the visibility window (i.e., $t_k^m$), and simultaneously forwarded to the next satellite via inter-satellite links to ensure the continuity of the learning process.

\subsubsection{\textbf{Local Model Training}} Both scheduled and unscheduled devices update their local models and complete the updates before the end of round $k$ (i.e., $t_{k+1}$). Since the initial models differ between these two types of devices, their local updates are described separately.

After receiving the updated global model $\boldsymbol{w}_k$, the scheduled device
$i$ performs $E_k^i$ local  epochs to update its local model as follows:
\begin{equation}
\boldsymbol{w}_k^i = \boldsymbol{w}_k - \eta \sum_{e=0}^{E_k^i-1}\nabla F_i(\boldsymbol{w}_{k,e}^i),
\quad \forall i\in\mathcal I_k,
\end{equation}
where $\eta$ is the learning rate and $\nabla F_i(\boldsymbol{w}_{k,e}^i)$ denotes the gradient of local loss
function $F_i(\cdot)$ with respect to model $\boldsymbol{w}_{k,e}^i$ after $e$ local epochs. Besides, (5) implies that
\begin{equation}
   \boldsymbol{w}_{k,0}^i = \boldsymbol{w}_k, \boldsymbol{w}_k^i = \boldsymbol{w}_{k,E_k^i}^i.
\end{equation}


For unscheduled devices $i$, the local model continues to be updated based on the local model $\boldsymbol{w}^i_{k-1-\tau^{i}_{k-1}}$ obtained in round $k-1$, which is trained based on global model
$\boldsymbol{w}_{k-1-\tau_{k-1}^i}$. The update details are below:

\begin{align}\label{eq:staleness1}
\boldsymbol{w}^i_{k-{\tau^i_k}} = \boldsymbol{w}^i_{k-1-\tau^i_{k-1}} \!\!- & \eta \!\!\sum_{e=0}^{E_k^i-1} \nabla F_i\left(\boldsymbol{w}^i_{k-1-\tau^i_{k-1},e}\right), \nonumber\\
&\qquad \qquad \qquad \qquad\forall i\in\mathcal{I}\setminus \mathcal{I}_{k}.
\end{align}
with the relation
\begin{equation}
    \boldsymbol{w}^i_{k-1-\tau^{i}_{k-1}} = \boldsymbol{w}_{k-1-\tau^{i}_{k-1},\hat{E}_{k-1}^i},
\end{equation}
where $\hat{E}^i_{k-1}$ is the cumulative number of local epochs from the start of round $k-1-{\tau^i_{k-1}}$  to the end of round $k-1$, denoted by:
\begin{equation}
\label{cumulative}
    \hat{E}^i_{k-1} = \sum_{j = k-1-{\tau^i_{k-1}}}^{k - 1} E_j^i, ~\forall i\in\mathcal{I}\setminus \mathcal{I}_{k}.
\end{equation}

From the above learning workflow, it is evident that both the selection of scheduled devices and the number of local epochs play critical roles in the global model update, thereby significantly affecting learning performance.
Accordingly, we optimize these factors to accelerate learning under visibility-window-related time constraints. To this end, we model the communication and computation delays in the next subsection.



\subsection{Communication and Computing Model}

For the uplink communication from ground devices to the LEO satellite, we consider a reserved channel allocation scheme, such as frequency-division multiple access with equal bandwidth allocation, under which the uplink transmission delay for scheduled device $i$ in round $k$ is denoted by $T_k^{i,\rm{U}}$.
For downlink transmission, the satellite broadcasts the global model to all scheduled devices, during which the achievable downlink rate is determined by the device with the worst channel condition and is expressed as
\begin{equation}
r_k^{i,\rm{D}}
=
\min_{i \in \mathcal{I}_k}
\left\{
B \log_2 \left(
1 + \frac{H_i P_s G_i G_s}{B N_0}
\right)
\right\}.
\end{equation}
Here, $B$ denotes the broadcast bandwidth, and $P_s$ is the satellite transmit power. The free-space channel gain is given by $H_i = (c_0 / 4\pi f_c d_i)^2$, where $c_0$ is the speed of light in vacuum, $f_c$ and $d_i$ are the carrier frequency and the distance between device $i$ and the satellite. $G_i$ and $G_s$ are the antenna gains of the device and the satellite, and $N_0$ is the noise power spectral density. Accordingly, the downlink transmission delay of device $i$ can be written as
\begin{equation}
T_k^{i,\rm{D}} = \frac{Z}{r_{k}^{i,\rm{D}}},
\end{equation}
where $Z$ represents the size of global model in bits.

For local model updates, let $\Phi$ denote the computation workload (in FLOPs) required to process one data sample in a local epoch, and let $f_i$ denote the computing capability of device $i$ (in FLOPs per second). Accordingly, the computation delay of device $i$ for $E_k^i$ local epochs in round $k$ can be expressed as
\begin{equation}
T_k^{i,\rm{C}} = \frac{E_k^i \, \Phi \, D_i}{f_i}.
\end{equation}

\section{Convergence Analysis and Problem Formulation}
In this section, we analyze the impact of device scheduling and the number of local epochs on the convergence of the proposed CoCoFL.
Motivated by the analytical results, we then formulate a joint optimization problem to accelerate convergence.

\subsection{Convergence Analysis}
Following \cite{Ma2021JSAC,Chen2025Tcom,Luo2022Infocom}, 
we adopt the following four assumptions to facilitate the convergence analysis. 

\textit{Assumption 1.} For any $\boldsymbol{w}_1, \boldsymbol{w}_2$, the local loss function $F_i(\cdot)$ is $L$-smooth:
\begin{equation}
F_i(\boldsymbol{w}_2) - F_i(\boldsymbol{w}_1) \leq \langle \nabla F_i(\boldsymbol{w}_1), \boldsymbol{w}_2 - \boldsymbol{w}_1 \rangle + \frac{L}{2} \| \boldsymbol{w}_2 - \boldsymbol{w}_1 \|^2. \nonumber
\end{equation}

\textit{Assumption 2.} For any $\boldsymbol{w}_1, \boldsymbol{w}_2$, the local loss function $F_i(\cdot)$ is $\mu$-strongly convex:
\begin{equation}
    F_i(\boldsymbol{w}_2) - F_i(\boldsymbol{w}_1) \geq \langle \nabla F_i(\boldsymbol{w}_1), \boldsymbol{w}_2 - \boldsymbol{w}_1 \rangle + \frac{\mu}{2} \| \boldsymbol{w}_2 - \boldsymbol{w}_1 \|^2.\nonumber
\end{equation}

\textit{Assumption 3.} 
For any model $\boldsymbol{w}$, the gradient of local loss function is uniformly bounded:
\begin{align*}
\| \nabla F_i(\boldsymbol{w}) \|^2 \leq G^2,
\end{align*}
where $G$ is a positive constant. 

\textit{Assumption 4.} The deviation between local and global gradients is uniformly bounded. Specifically, for any model $\boldsymbol{w}$, there exists a constant $\xi > 0$ such that
\begin{align*}
\left\| \nabla F_i(\boldsymbol{w}) - \nabla F(\boldsymbol{w}) \right\|^2 \leq \xi^2.
\end{align*}




\begin{theorem}
Under Assumptions 1-4 and the condition that the learning rate satisfies $\eta < \mu^2 / L^3$, we have
\begin{equation}
\label{theorem1-1}
\mathbb{E}[F(\boldsymbol{w}_k) - F(\boldsymbol{w}^*)] \leq \rho^k \cdot \mathbb{E}[F(\boldsymbol{w}_0) - F(\boldsymbol{w}^*)] + \frac{R_{\max}}{1 - \rho}, 
\end{equation}
where $R_{\max}$ is a constant and $\rho \triangleq \max_k \rho_k$ lies in the range of $(0,1)$. Besides, $\rho_k$ is defined as
\begin{equation}
    \rho_k = 1 - 2\mu\eta\sum_{i \in \mathcal{I}} {\color{black}\mathbb{E}[s_k^i]} \beta_i \left(\hat{E^i_k}-\frac{\eta L^3}{\mu^2} (\hat{E^i_k})^2\right),
    \label{eq: rho_def}
\end{equation}
where $s_k^i \in \{0,1\}$ indicates whether device $i$ is scheduled in round $k$ (i.e., $s_k=1$) or not (i.e., $s_k=0$).

\end{theorem}

\begin{proof}
The proof is omitted due to space limitations.
\end{proof}

From Theorem 1, we derive insights for optimizing device scheduling and the number of local epochs, as summarized in the following remark.

\begin{remark}
(\ref{theorem1-1}) and (\ref{eq: rho_def}) indicate that decreasing $\rho_k$ in round $k$ helps improve both the convergence rate and accuracy. Under the assumption that each device is scheduled independently and identically distributed in each round, minimizing $\rho_k$ is equivalent to maximize the following term:
\begin{equation}
\label{eq:convergence factor}
    \tilde{\rho}_k = \sum_{i\in\mathcal{I}}s_k^i\beta_i\left(\hat{E^i_k}- a (\hat{E^i_k})^2\right),
\end{equation}
\vspace*{0.01in}

\noindent
with $ 0 <\tilde{\rho}_k < \frac{1}{2\mu\eta}$ and $a=\frac{\eta L^3}{\mu^2}<1$. Under the condition that $\tilde{\rho}_k>0$, we derive
\begin{equation}
    1\leq \hat{E^i_k} \leq \frac{1}{a}-1.
    \label{eq:condition_1}
\end{equation}
To ensure $\rho_k<1$, we have
\begin{equation}
    \sum_{i\in\mathcal{I}} s_k^i\beta_i\left(\hat{E^i_k}-a (\hat{E^i_k})^2\right) \leq b,
    \label{eq:condition_2}
\end{equation}
where $b=\frac{1}{2\mu\eta}-\epsilon$, with $\epsilon$ denoting a small positive constant.


For learning acceleration, (\ref{eq:convergence factor}), (\ref{eq:condition_1}), and (\ref{eq:condition_2}) jointly provide guidance for formulating the optimization objective and constraints.
In detail, (\ref{eq:convergence factor}) as an objective function indicates that, under communication constraints, scheduling more devices—particularly those with larger relative dataset sizes $\beta_i$—generally enhances the convergence performance. Furthermore, the quadratic dependence on the cumulative number of local epochs $\hat{E}_k^i$ in (\ref{eq:convergence factor}), together with the individual and total constraints in (\ref{eq:condition_1}) and (\ref{eq:condition_2}), underscores the importance of appropriately selecting devices based on their $\hat{E}_k^i$. If $\hat{E}_k^i$ is either too large or too small, the device should be assigned a low scheduling priority, due to either staleness information or insufficient learning in its local model.

\end{remark}

\subsection{Problem Formulation}
Guided by the convergence analysis, we formulate the following optimization problem:
\begin{equation*}
\begin{aligned}
\text{(P0)} ~\underset{ \mathcal{S}, \mathcal{E}}{\text{max}} & \; \sum_{i \in \mathcal{I}}  s^i_k \beta_i \left((\hat{E}^i_{k-1}+E^i_k) - a(\hat{E}^i_{k-1}+E^i_k)^2\right) \\
\text{s.t. } & \; \text{C1: } t_{k} + T^{i,\rm{C}}_k + {s}^i_k(T^{i,\rm{U}}_k + T^{i,\rm{D}}_k)\leq t_{k+1}, \forall i \in \mathcal{I}\\
& \; \text{C2: } t_{k} + {s}^i_k(T^{i,\rm{U}}_k + T^{i,\rm{D}}_k)\leq t^m_k, \forall i \in \mathcal{I}\\
& \; \text{C3: } \!\!\sum_{i\in\mathcal{I}} s_k^i\beta_i((\hat{E}^i_{k-1}+E^i_k)-a(\hat{E}^i_{k-1}+E^i_k)^2) \leq b\\
& \; \text{C4: } 1\leq \hat{E}^i_{k-1}+E^i_k \leq \frac{1}{a}-1,\forall i \in \mathcal{I}\\
& \; \text{C5: } {s}^i_k \in \{0, 1\},\forall i \in \mathcal{I}\\
& \; \text{C6: } {E}^i_k \in \mathbb{Z},\forall i \in \mathcal{I},
\end{aligned}
\end{equation*}
where $\mathcal{S}$ and $\mathcal{E}$ are the optimal variables, indicating the device scheduling and the number of local epochs, respectively. The objective function is derived from \eqref{eq:convergence factor}. 
Moreover, C1 guarantees that the total computation and communication time incurred by each device in a round does not exceed the round duration. C2 ensures that all communications occur within the visibility window. C3 and C4 are from \eqref{eq:condition_1} and \eqref{eq:condition_2}, jointly bounding the cumulative number of local epochs of each device.
C5 defines the binary scheduling decision and C6 enforces that the number of local epochs for each device is an integer. Note that in problem (P0), both $a$ and $b$ are constants determined through  experiments.

\section{Proposed Learning Acceleration Algorithm}
To make the nonlinear mixed-integer problem (P0) tractable, we propose a two-level algorithm that
integrates an outer Gibbs sampling for device scheduling with an inner difference of convex functions algorithm (DCA) to determine the number of local epochs. The outer and inner optimization procedures are detailed in the sequel, respectively.


\subsection{Inner Optimization via DCA}
\label{subsec:other_algo1}
With fixed device scheduling, the number of local epochs for both scheduled and unscheduled devices is optimized in this subsection. First, each local epoch $E_k^i$ is relaxed to be continuous (i.e., C6 is omitted) for tractability. Then, the number of local epochs for the two types of devices is optimized separately.

For unscheduled devices, no communication delay occurs, and their local epochs are constrained by C1 and C4 from problem (P0). The feasible range of local epochs for each unscheduled device $i$ is thus given by
\begin{equation}
E_k^{i} \in \left[0, \min\Bigg\{ \frac{(t_{k+1}-t_k) f_i}{ \Phi D_i}, \frac{1}{a}-1-\hat{E}_{k-1}^i \Bigg\}\right].
\end{equation}
Each unscheduled device performs the maximum feasible number of local epochs within this range, to align with the continual computing feature of the proposed CoCoFL. That is, 
\begin{equation}
\label{unscheduled loacl epochs}
    E_k^i =\min\left\{ \frac{(t_{k+1}-t_k) f_i}{ \Phi D_i}, \frac{1}{a}-1-\hat{E}_{k-1}^i \right\}.
\end{equation}


For the scheduled devices, the local epochs are optimized using DCA. With fixed device scheduling and omitted C6, the subproblem for these devices can be reformulated as follows: 
\begin{equation*}
\begin{aligned}
\text{(P1)} ~\underset{\mathcal{E}}{\text{max}} & \; \sum_{i \in \mathcal{I}_k}   \beta_i \left((\hat{E}^i_{k-1}+E^i_k) - a(\hat{E}^i_{k-1}+E^i_k)^2\right) \\
\text{s.t. } & \; \text{C7: } t_{k} + T^{i,\rm{C}}_k + T^{i,\rm{U}}_k + T^{i,\rm{D}}_k\leq t_{k+1}, \forall i \in \mathcal{I}_k\\
& \; \text{C8: } t_k + T^{i,\rm{U}}_k + T^{i,\rm{D}}_k\leq t^m_k, \forall i \in \mathcal{I}_k\\
& \; \text{C9: } \sum_{i\in\mathcal{I}_k} \beta_i((\hat{E}^i_{k-1}+E^i_k)-a(\hat{E}^i_{k-1}+E^i_k)^2) \leq b\\
& \; \text{C10: } 1\leq \hat{E}^i_{k-1}+E^i_k \leq \frac{1}{a}-1,\forall i \in \mathcal{I}_k,
\end{aligned}
\end{equation*}
where C8 is independent of the local epochs and thus serves as a feasibility check for the device scheduling.
Specifically, if C8 is not satisfied, the current device scheduling is infeasible; otherwise, it is feasible and problem (P1) with C7, C9, and C10 is solved to determine the number of local epochs for the scheduled devices.
C9 is non-convex but is in the form of a difference of convex (DC) function. Hence, we linearize the non-convex component in C9 to obtain its convex approximation as
\begin{equation}
\begin{aligned}
\sum_{i\in\mathcal{I}_k} \beta_i (\hat{E}^i_{k-1} + E^i_k)
- \sum_{i\in\mathcal{I}_k} a \beta_i \Big[
    (\hat{E}^i_{k-1} + E^{i,r}_k)^2  
+  \\
\qquad 2(\hat{E}^i_{k-1} + E^{i,r}_k)(E^i_k - E^{i,r}_k)
\Big]
\le b.
\end{aligned}
\label{eq:DCA6}
\end{equation}
By employing the DCA, problem (P1) with C7, C10, and \eqref{eq:DCA6} is solved iteratively until convergence. Here, $E_k^{i,r}$ in \eqref{eq:DCA6} denotes the solution obtained in the $r-1$-th iteration.


\subsection{Outer Optimization via Gibbs Sampling}
This subsection introduces a Gibbs sampling approach for device scheduling \cite{XuefeiTMC}.
At the $t$-th sampling, one of three operations—\textit{add}, \textit{remove}, or \textit{swap}—is randomly selected to generate a candidate schedule $\tilde{\mathcal{S}}_k^t$ from the current schedule $\mathcal{S}_k^t$. For this candidate schedule, the local epochs of unscheduled devices are determined according to \eqref{unscheduled loacl epochs}, while the local epochs of scheduled devices are obtained by solving problem (P1). If C8 is not satisfied, the candidate schedule $\tilde{\mathcal{S}}_k^t$ is discarded and the next sampling iteration is performed. Otherwise, the candidate objective value $\tilde{{obj}^t}$ is obtained by solving problem (P1) without C8 using the DCA, and the difference between the current and candidate objective values, $\Delta obj = {obj}^t - \tilde{{obj}^t}$, is calculated to determine the acceptance probability of the candidate schedule
\begin{equation}
\label{p_acc}
p_{\rm acc} = \frac{1}{1 + \exp(\Delta obj / \gamma)},
\end{equation}
where $\gamma > 0$ is a temperature parameter that controls the exploration–exploitation trade-off. 
With probability $p_{\rm acc}$, the candidate schedule is accepted, i.e., $\mathcal{S}_k^{t+1} = \tilde{\mathcal{S}}_k^t$ and $obj^{t+1} = \tilde{{obj}^t}$; otherwise, the current schedule is retained. 
After $T$ samplings, the obtained device schedule $\mathcal{S}_k^T$ and local epochs $\mathcal{E}_k^T$ are output as the optimized parameters for the workflow. The cumulative number of local epochs $\hat{\mathcal{E}}$ is set to $\mathcal{E}_k^T$ for scheduled devices, while for unscheduled devices it is set following \eqref{cumulative}.

For clarify, we summarize the complete CoCoFL optimization algorithm in Algorithm 1, where the resulting device scheduling and the number of local epochs are incorporated into the workflow in Section \ref{sub:CoCoFL} for learning acceleration.



\addtolength{\topmargin}{0.05in}
\begin{algorithm}[htbp]
\caption{CoCoFL optimization algorithm}
\label{alg:overall}
\begin{algorithmic}[1]
\STATE Initialize cumulative number of local epochs $\hat{\mathcal{E}} \gets \mathbf{0}$

\FOR{round $k = 1$ \TO $K$}    
    \STATE{Initialize device schedule $\mathcal{S}_k^1$, local epochs $\mathcal{E}_k^1$, and objective value ${obj}^1$}
    \FOR{sampling $t = 1$ \TO $T$}
        \STATE Generate candidate device schedule $\tilde{\mathcal{S}}_k^t$ based on $\mathcal{S}_k^t$
        \STATE For unscheduled devices, determine the number of local epochs according to \eqref{unscheduled loacl epochs}
        \IF{C8 is satisfied}
        \STATE For scheduled devices in $\tilde{\mathcal{S}}_k^t$, solve problem \text{(P1)} regardless of C8 to obtain the number of local epochs $\tilde{\mathcal{E}}_k^t$ and objective value $\tilde{{obj}^t}$ by DCA
        \STATE Update $\mathcal{S}_k^t=\tilde{\mathcal{S}}_k^t$, $\mathcal{E}_k^t=\tilde{\mathcal{E}}_k^t$, $obj^t=\tilde{obj^t}$ with probability $p_{acc}$ as defined in \eqref{p_acc}
        \ENDIF
    \ENDFOR
\STATE Output $\mathcal{S}_k^T$ and ${\mathcal{E}}_k^T$ rounded down to the nearest integer and incorporate them into the workflow.
\STATE Update $\hat{\mathcal{E}}$ for the optimization in the next round
\ENDFOR
\end{algorithmic}
\end{algorithm}

\section{Performance Evaluation} 
In this section, we first give the experimental setup and then evaluate the performance of the proposed CoCoFL optimization algorithm against representative baselines.

\subsection{Experimental Setup}

\subsubsection{System Settings}
The system comprises 40 ground devices distributed near Singapore (approximately $1^{\circ}$N, $104^{\circ}$E) and a single-plane Walker–Delta constellation of 12 LEO satellites at an altitude of 600 km. The visibility windows of the satellites are simulated using the Satellite Tool Kit (STK).

The communication parameters are configured as follows: the carrier
frequency is $f_c = 2~\mathrm{GHz}$, the total available bandwidth is
$B = 20~\mathrm{MHz}$, and the noise power spectral
density is $N_0 = -174~\mathrm{dBm/Hz}$ \cite{3gpp38811}. The transmit antenna gain of each device is $G_i = 4~\mathrm{dBi}$, while the receive antenna gain of the
satellite is $G_s = 35~\mathrm{dBi}$. The transmit power
of each device is independently and uniformly distributed within
$[0.01, 0.1]~\mathrm{W}$, whereas the satellite transmission power is fixed at $P_s = 50~\mathrm{W}$. The uplink transmission delay can be determined from the uplink transmission rate, which is calculated using the Shannon formula under the free-space channel model with equal bandwidth allocation among the scheduled devices.

Regarding local model updates, the computing capability of each device is set to $f_i = 4.8~\mathrm{GFLOPs/s}$. The Fashion-MNIST dataset, comprising $60,000$ samples, is adopted and distributed across all devices, with local datasets heterogeneous in both size and class distribution. Each device adopts the VGG-11 model, with a model size of $Z = 108~\mathrm{MB}$, and trains it using the stochastic gradient descent (SGD) optimizer. Under this model, the per-sample computation workload is set
to $\Phi = 327$~MFLOPs for Fashion-MNIST.

\begin{figure*}[htbp]
\centering
\begin{minipage}[b]{0.32\textwidth}
    \centering
    \includegraphics[width=\textwidth]{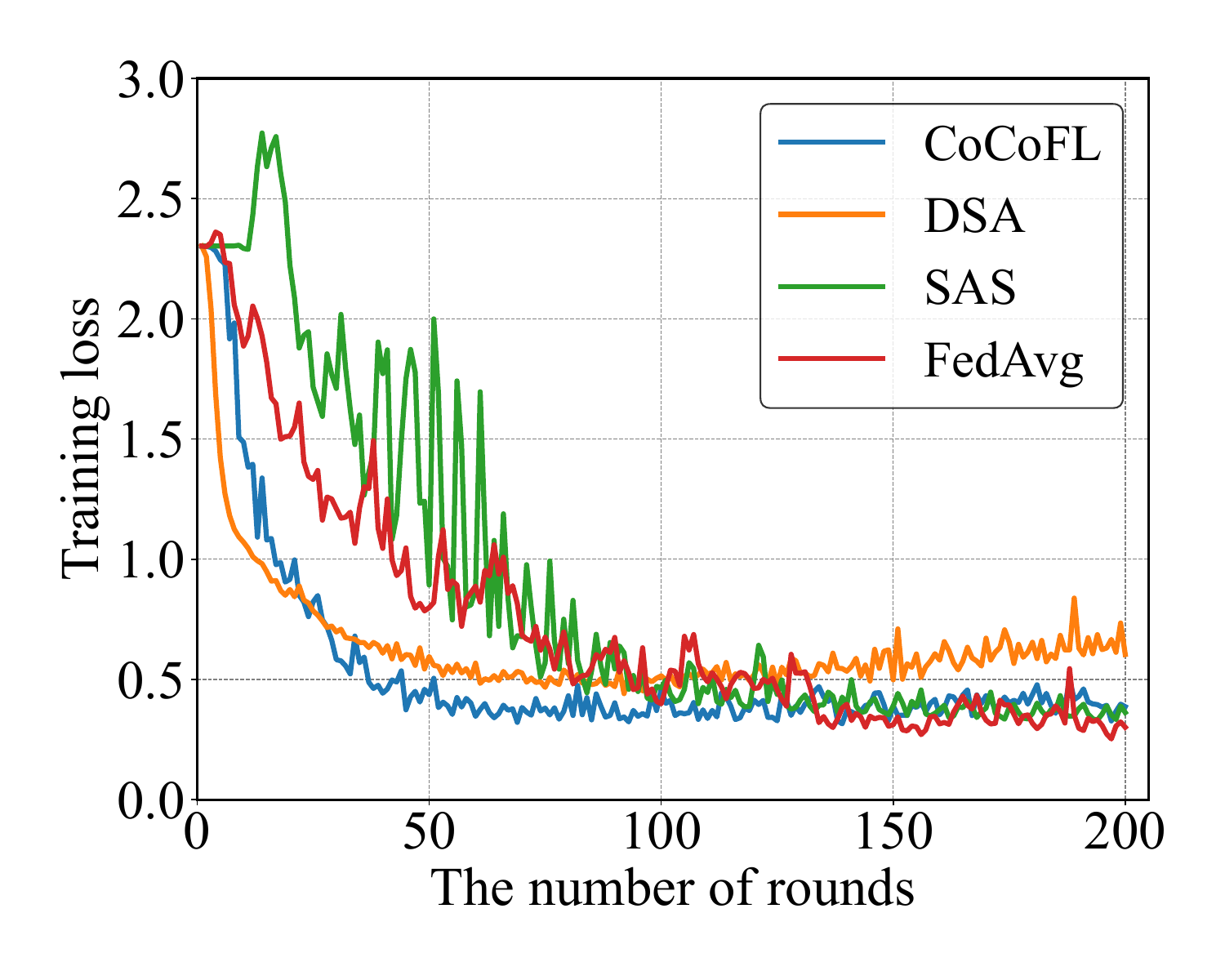}
    \vspace{-15pt}  
    \caption{Training loss vs. the number of rounds.}
    \label{fig:loss}
\end{minipage}
\vspace{-0.1cm}
\hfill
\begin{minipage}[b]{0.32\textwidth}
    \centering
    \includegraphics[width=\textwidth]{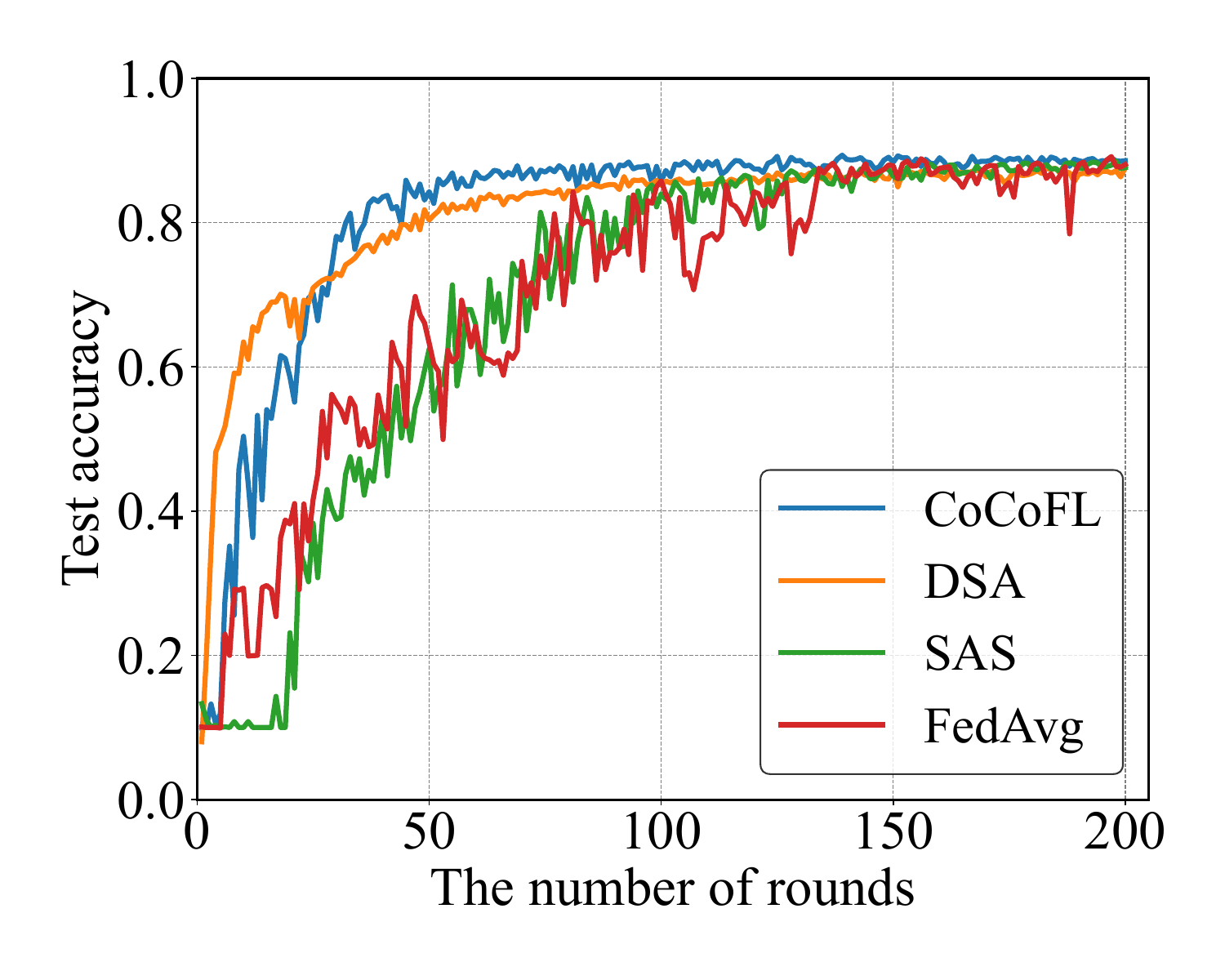}
    \vspace{-15pt} 
    \caption{Test accuracy vs. the number of rounds.}
    \label{fig:acc}
\end{minipage}
\vspace{-0.1cm}
\hfill
\begin{minipage}[b]{0.32\textwidth}
    \centering
    \includegraphics[width=\textwidth]{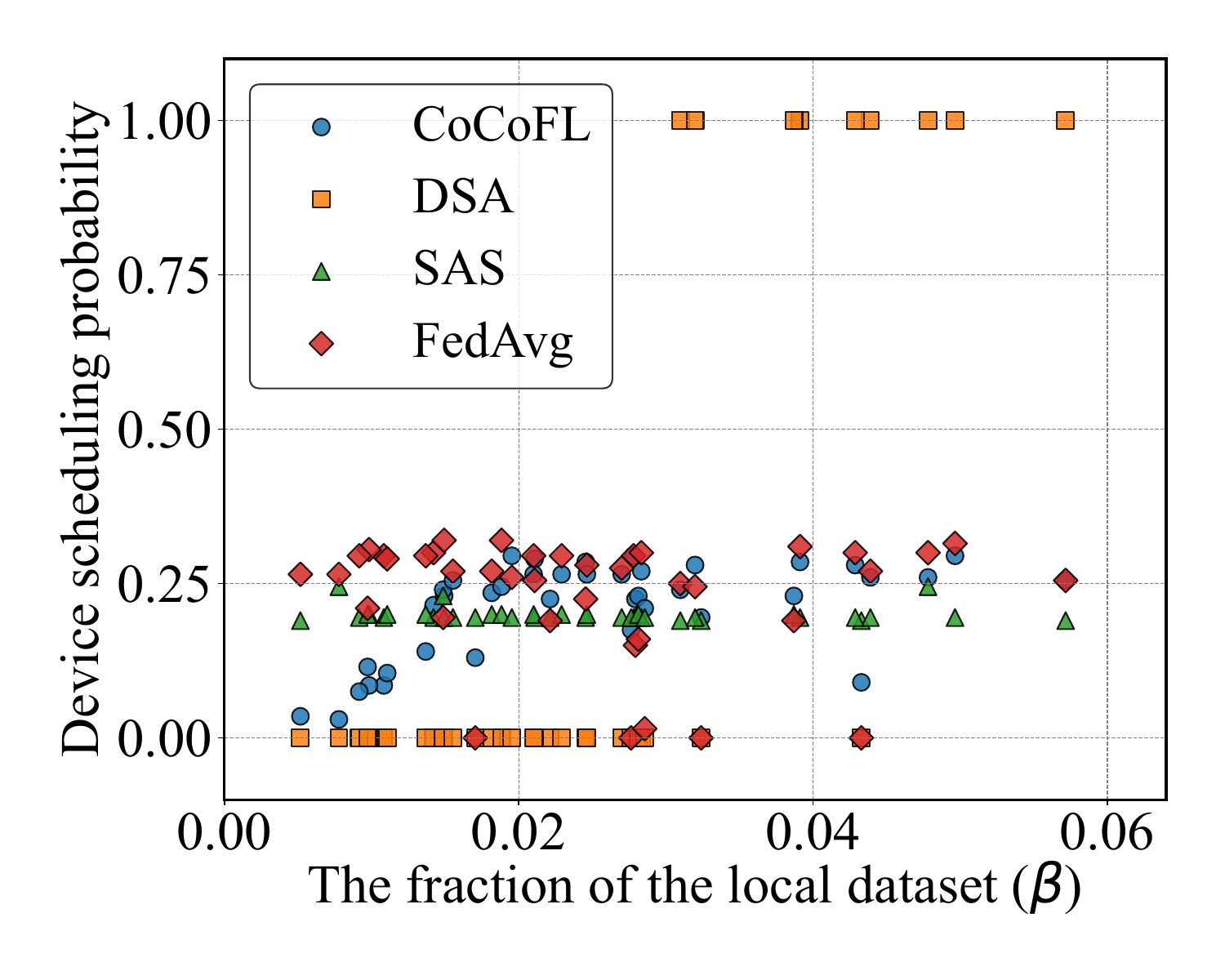}
    \vspace{-15pt} 
    \caption{Device scheduling probability.} 
    \label{fig:scheduling}
\end{minipage}
\vspace{-0.1cm}
\end{figure*}

\subsubsection{Benchmark Algorithms}
The following algorithms are considered as the baselines:

\begin{enumerate}
\item \textbf{DSA (Data-Size-Aware Device Scheduling \cite{Guo2022JSTSP})}: 
In each round, devices are sorted in descending order of local dataset size and scheduled sequentially until the communication delay of any scheduled device would exceed the satellite visibility window (i.e., C2 is not satisfied). The number of local epochs is set to the maximum value restricted by C1.


\item \textbf{SAS (Staleness-Aware Scheduling\cite{Yang2020ICASSP})}: This baseline employs Gibbs sampling to maximize the total staleness coefficient of scheduled devices subject to C2. Local epoch settings follow the same rule as in DSA.

\item \textbf{FedAvg (Federated Averaging \cite{mcmahan2017communication})}: FedAvg adopts a greedy random scheduling strategy, where devices are sequentially and randomly added to the scheduled set until C2 is violated. Local epoch settings follow the same rule as in DSA. For a fair comparison, the global model aggregation in FedAvg is revised following \eqref{aggregate}.

\end{enumerate}

\subsection{Experimental Results}

Figs.~\ref{fig:loss} and~\ref{fig:acc} show the training loss and test accuracy comparisons.
The results demonstrate that the proposed CoCoFL optimization algorithm, hereafter referred to as CoCoFL, achieves significantly faster convergence than FedAvg and SAS. The improvement over FedAvg arises from its continual computing mechanism, which helps incorporate more knowledge into the local model. Compared with SAS, CoCoFL additionally prioritizes devices with larger datasets alongside staleness consideration, ensuring that more informative updates contribute to the global model. Notably, although CoCoFL exhibits slightly slower learning improvement than DSA in early rounds, it attains lower training loss and higher test accuracy in subsequent rounds. This is because DSA always schedules the devices with larger datasets, which drives the converged model away from devices with smaller datasets and results in poor global performance.

Fig.~\ref{fig:scheduling} illustrates the scheduling probability of each device, which further reveals the source of performance gain achieved by the proposed CoCoFL. It is observed that
DSA exclusively selects devices with larger datasets while neglecting those with smaller datasets.
FedAvg schedules devices randomly, while SAS selects devices in a round-robin manner to maintain low staleness coefficients; in both cases, dataset size is disregarded in the device scheduling process.
Consequently, in FedAvg and SAS, the scheduling probabilities across devices with different dataset sizes are approximately uniform. In contrast, the proposed CoCoFL primarily prioritizes devices with larger datasets while still assigns a low scheduling probability to devices with smaller datasets, thereby promoting faster convergence and improving both training loss and test accuracy.

\section{Conclusions} 
In this paper, we have proposed CoCoFL, a continual computing based federated learning framework over intermittent satellite–ground links. 
By parallelizing local model updates of unscheduled devices and the global model aggregation of scheduled devices, CoCoFL fully exploits local computational and data sizes to improve learning performance. In addition, guided by the convergence analysis, we have developed a two-level optimization algorithm within CoCoFL to jointly optimize device scheduling and the number of local epochs for both scheduled and unscheduled devices.  Experimental results demonstrate that CoCoFL achieves superior learning performance improvements compared with baselines, in terms of faster convergence, lower training loss, and higher accuracy.

\normalem
\bibliographystyle{IEEEtran}
\bibliography{ref}
\end{document}